\documentclass[preprint,10pt]{elsarticle}

\usepackage[a4paper,top=4.0cm,bottom=4.0cm]{geometry}

\usepackage{amsmath,amssymb,amsthm}

\usepackage{graphicx}
\usepackage{subcaption}
\usepackage{tikz}
\usetikzlibrary{shapes,arrows.meta,calc,positioning,decorations.markings}

\usepackage{hyperref}
\hypersetup{colorlinks=true, urlcolor=red, citecolor=blue, linkcolor=red}

\makeatletter

\newcommand*\MyORCID{0009-0004-3236-9811}
\newcommand*\MyORCIDLink{%
  \href{https://orcid.org/\MyORCID}{\MyORCID}%
}

\gdef\emailauthor#1#2{%
  \stepcounter{ead}%
  \g@addto@macro\@elseads{%
    \raggedright
    \let\corref\@gobble
    \def\@@tmp{#1}%
    \eadsep{\ttfamily\expandafter\strip@prefix\meaning\@@tmp}%
    \space(#2, ORCID:\space\MyORCIDLink)%
    \def\eadsep{\unskip,\space}%
  }%
}

\makeatother

\newtheorem{theorem}{Theorem}
\newtheorem{lemma}[theorem]{Lemma}

\begin{document}

\begin{frontmatter}

\title{A Structural Equivalence of Symmetric TSP to a Constrained Group Steiner Tree Problem}

\author{Y\i lmaz Arslano\u{g}lu}
\ead{yilmaz@hoketo.com}

\affiliation{organization={Independent Researcher}, country={Hamburg, Germany}}

\begin{abstract}
We present a brief structural equivalence between the symmetric TSP and a constrained Group Steiner Tree Problem (cGSTP) defined on a simplicial incidence graph. Given the complete weighted graph on the city set V, we form the bipartite incidence graph between triangles and edges. Selecting an admissible, disk-like set of triangles induces a unique boundary cycle. With global connectivity and local regularity constraints, maximizing net weight in the cGSTP is exactly equivalent to minimizing the TSP tour length.
\end{abstract}


\end{frontmatter}

\section{Introduction and Geometric View}

Let $V$ be a set of $n\ge 3$ cities and $G=(V,E)$ be the complete undirected graph with symmetric edge lengths $L_e>0$. The \emph{symmetric Traveling Salesman Problem (TSP)} asks for a minimum-length Hamiltonian cycle on $V$ \cite{applegate2006}.

In this note, we describe an exact structural equivalence of symmetric TSP to a constrained Group Steiner Tree Problem (GSTP) \cite{reich1990}.
The construction utilizes elementary combinatorial topology \cite{edelsbrunner2010} to view a tour not as a one-dimensional cycle, but as the \emph{boundary} of a two-dimensional simplicial surface.
By selecting a set of triangles forming a topological disk, edges shared by two selected triangles (internal edges) cancel out, leaving a unique boundary cycle.
Figure~\ref{fig:toy_instance} shows the same surface--boundary viewpoint on a small symmetric instance:
a connected selection of triangles behaves as an abstract disk, and its induced boundary is a single cycle (the tour).

\begin{figure}[t!]
\centering
\includegraphics[width=1.0\linewidth]{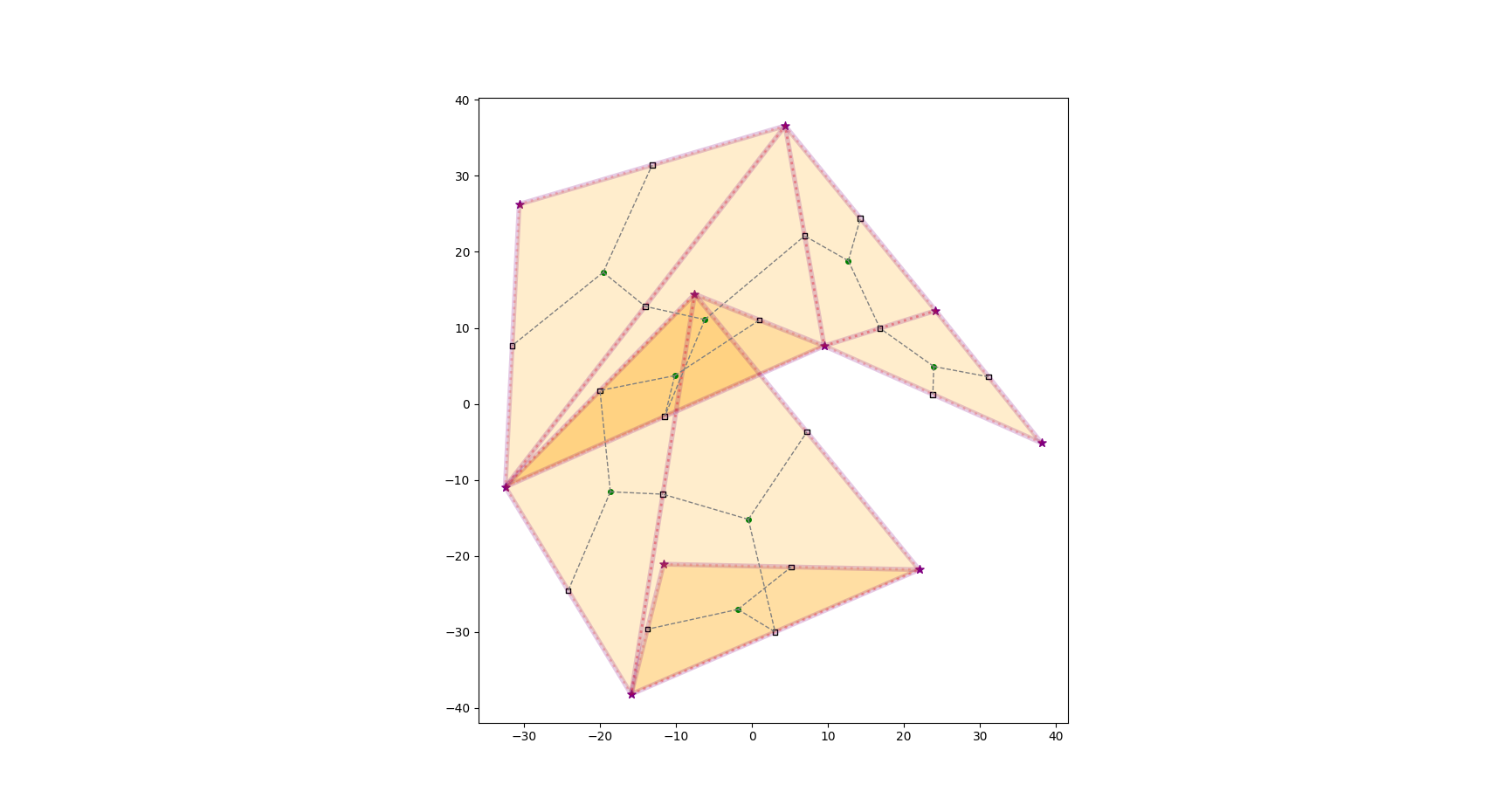}
\caption{A small symmetric instance ($n=10$) shown with an arbitrary planar embedding for visualization.
A connected set of selected triangles (shaded) forms an abstract disk; internal edges cancel, and the induced boundary
is a single simple cycle (the tour).}
\label{fig:toy_instance}
\end{figure}

\section{The cGSTP Equivalence}

We define the cGSTP on the bipartite \emph{incidence graph} $B=(U\cup W, A)$ (see Fig.~\ref{fig:incidence_graph}).
$U=\binom{V}{3}$ is the set of triangle nodes (circles), $W=E$ is the set of primal edge nodes (squares), and $(t,e)\in A$ denotes incidence.
For each city $v\in V$, we define a group $U(v)=\{t\in U: v\in t\}$.

Triangle nodes act as \emph{group terminals} (cost 0), while edge nodes and incidences carry weights defined to enforce boundary cancellation:
\begin{enumerate}
    \item Each edge node $e \in W$ has cost $c(e) = 2L_e$.
    \item Each incidence arc $(t,e) \in A$ has profit $p(t,e) = L_e$.
\end{enumerate}

Let $x_t, y_e, z_{t,e} \in \{0,1\}$ indicate the selection of triangles, edges, and incidences respectively.
We select an ``admissible'' subgraph $B'$ to maximize the net weight:
\[
W(B') := \sum_{(t,e)\in A} z_{t,e} L_e \;-\; \sum_{e\in W} y_e 2L_e.
\]
Admissibility is defined by the following constraints:

\paragraph{(C1) Incidence linking (Terminal Node Degrees)} Every active triangle must use all three of its edges:
\[
z_{t,e} \le y_e
\]
\[
z_{t,e} \le x_t
\]
\[
\sum_{e\subset t} z_{t,e} = 3x_t
\]

\paragraph{(C2) Manifold regularity (Steiner Node Degrees)} Every primal edge is incident to at most two selected triangles (Fig.~\ref{fig:constraints_viz}A):
\[
y_e \le \sum_{t\supset e} z_{t,e} \le 2y_e
\]

\paragraph{(C3) Global Euler counts (Node Cardinalities)} These match the cardinality of an abstract triangulated disk with $n$ vertices:
\[
\sum x_t = n-2
\]
\[
\sum y_e = 2n-3
\]

\begin{figure}[t!]
\centering
\begin{tikzpicture}[scale=1.25, thick]
    \tikzset{
        city/.style={circle,fill=black,inner sep=1.5pt,outer sep=1pt},
        tri_node/.style={circle,draw=orange!80!black,fill=orange!30,inner sep=2pt},
        edge_node/.style={rectangle,draw=black,fill=white,inner sep=2.5pt},
        dual_edge/.style={draw=blue!80!black,dashed,line width=0.55pt},
        selected_dual/.style={draw=green!60!black,line width=1.4pt,opacity=0.55},
        tri_fill/.style={fill=orange!25,draw=none,opacity=0.28},
        boundary_edge/.style={draw=red!85!black,ultra thick}
    }
    \def\missingTri{1}
    \coordinate (C) at (0,0);
    \foreach \i in {1,...,6} \coordinate (V\i) at ({(\i-1)*60 + 90}:1.75);

    \foreach \i [evaluate=\i as \next using {int(mod(\i,6)+1)}] in {1,...,6} {
        \ifnum\i=\missingTri \else \path[tri_fill] (C) -- (V\i) -- (V\next) -- cycle; \fi
    }

    \draw[gray!30, line width=0.45pt] (C) -- (V1) (C) -- (V2) (C) -- (V3) (C) -- (V4) (C) -- (V5) (C) -- (V6);
    \draw[gray!30, line width=0.45pt] (V1) -- (V2) -- (V3) -- (V4) -- (V5) -- (V6) -- cycle;

    \foreach \i [evaluate=\i as \next using {int(mod(\i,6)+1)}] in {1,...,6} {
        \ifnum\i=\missingTri \else \draw[boundary_edge] (V\i) -- (V\next); \fi
    }
    \pgfmathtruncatemacro{\nextMissing}{int(mod(\missingTri,6)+1)}
    \draw[boundary_edge] (C) -- (V\missingTri); \draw[boundary_edge] (C) -- (V\nextMissing);

    \node[city, label={[font=\small, yshift=-2pt]below:$v$}] (c_center) at (C) {};
    \foreach \i in {1,...,6} \node[city] at (V\i) {};
    \foreach \i [evaluate=\i as \next using {int(mod(\i,6)+1)}] in {1,...,6}
        \node[tri_node] (t\i) at (barycentric cs:C=1,V\i=1,V\next=1) {};
    \foreach \i in {1,...,6}
        \node[edge_node] (e_in\i) at (barycentric cs:C=1,V\i=1) {};
    \foreach \i [evaluate=\i as \next using {int(mod(\i,6)+1)}] in {1,...,6}
        \node[edge_node] (e_out\i) at (barycentric cs:V\i=1,V\next=1) {};

    \foreach \i [evaluate=\i as \next using {int(mod(\i,6)+1)}] in {1,...,6} {
        \draw[dual_edge] (t\i) -- (e_in\i); \draw[dual_edge] (t\i) -- (e_in\next); \draw[dual_edge] (t\i) -- (e_out\i);
    }
    \foreach \i [evaluate=\i as \next using {int(mod(\i,6)+1)}] in {1,...,6} {
        \ifnum\i=\missingTri \else
            \draw[selected_dual] (t\i) -- (e_out\i);
            \draw[selected_dual] (t\i) -- (e_in\i);
            \draw[selected_dual] (t\i) -- (e_in\next);
        \fi
    }

    \begin{scope}[shift={(2.85, 0)}]
        \node[tri_node] at (0, 1.15) {};
        \node[right, font=\scriptsize, align=left] at (0.28, 1.15) {\textbf{Triangle node} ($t \in U$)\\Cost $0$};
        \node[edge_node] at (0, 0.55) {};
        \node[right, font=\scriptsize, align=left] at (0.28, 0.55) {\textbf{Edge node} ($e \in W$)\\Cost $2L_e$};
        \draw[dual_edge] (0, 0.02) -- (0.45, 0.02);
        \node[right, font=\scriptsize, align=left] at (0.55, 0.02) {Incidence\\Profit $L_e$};
        \draw[selected_dual] (0, -0.40) -- (0.45, -0.40);
        \node[right, font=\scriptsize, align=left] at (0.55, -0.40) {Selected in tree};
        \draw[boundary_edge] (0, -0.82) -- (0.45, -0.82);
        \node[right, font=\scriptsize, align=left] at (0.55, -0.82) {Induced boundary};
    \end{scope}
\end{tikzpicture}
\caption{The incidence graph structure. Selecting a set of triangles that form a topological disk (shaded) results in cost cancellation on internal edges, leaving a net cost corresponding to the induced boundary tour (red).}
\label{fig:incidence_graph}
\end{figure}
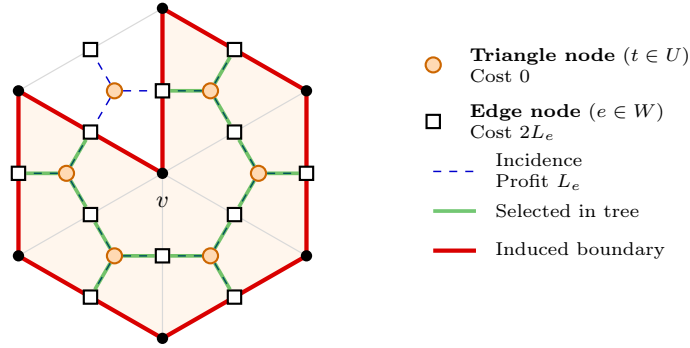

\paragraph{(C4) Global connectivity (Tree)} The active subgraph $B'$ must be a tree.

\paragraph{(C5) Local connectivity (Steiner Groups)}
For each city $v$, the active local incidence subgraph $H_v$ must satisfy the \textit{Euler Characteristic}:
\[
\chi(H_v)=|V(H_v)|-|E(H_v)|=1
\]
Here, $H_v$ is the subgraph of $B'$ induced by selected triangles and edges incident to vertex $v$ (cf. the \textit{star} of $v$ in Fig.~\ref{fig:incidence_graph}).
Given the global tree constraint, this enforces the vertex link to be a simple path
(excluding bowties like Fig.~\ref{fig:constraints_viz}B and cyclic vertex links, which would arise if the local fan in Fig.~\ref{fig:incidence_graph} were closed).
The presence of at least one terminal per group is thus also implicitly enforced.

\begin{figure}[t!]
\centering
\begin{subfigure}[b]{0.45\textwidth}
    \centering
    \begin{tikzpicture}[scale=0.9, line join=round]
        \node[anchor=south, font=\bfseries\footnotesize] at (0, 3.05) {A. Non-manifold edge};
        \coordinate (B) at (0, 0); \coordinate (T) at (0, 3);
        \filldraw[fill=orange!20, draw=orange!80!black, thin, opacity=0.8] (B)--(T)--(2, 1.5)--cycle;
        \filldraw[fill=orange!30, draw=orange!80!black, thin, opacity=0.8] (B)--(T)--(-1.5, 2)--cycle;
        \filldraw[fill=orange!50, draw=orange!80!black, thin, opacity=0.8] (B)--(T)--(-1.2, 0.5)--cycle;
        \draw[ultra thick, black] (B) -- (T);
        \node[rectangle, draw=black, fill=white, inner sep=2pt] at (0, 1.5) {};
    \end{tikzpicture}
\end{subfigure}
\hfill
\begin{subfigure}[b]{0.45\textwidth}
    \centering
    \begin{tikzpicture}[scale=0.7]
        \node[anchor=south, font=\bfseries\footnotesize] at (0, 3.05) {B. ``Bowtie'' singularity};

        \coordinate (C)  at (0,1.5);
        \coordinate (L1) at (-1.5,2.5);
        \coordinate (L2) at (-1.5,0.5);
        \coordinate (R1) at ( 1.5,2.5);
        \coordinate (R2) at ( 1.5,0.5);

        \fill[orange!20] (C) -- (L1) -- (L2) -- cycle;
        \fill[orange!20] (C) -- (R1) -- (R2) -- cycle;

        \fill[red] (C) circle (3pt);
        \node[right, red, font=\bfseries] at (0.1, 1.5) {$v$};

        \draw[magenta, thick] (C) -- (L1) -- (L2) -- cycle;
        \draw[magenta, thick] (C) -- (R2) -- (R1) -- cycle;
    \end{tikzpicture}
\end{subfigure}

\caption{Forbidden anomalies: (A) Edge incident to $>2$ triangles; (B) Disconnected vertex link.}
\label{fig:constraints_viz}
\end{figure}
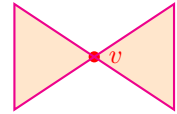

\section{Theoretical Properties}

Let $K=\{t\in U : x_t=1\}$ be the selected triangles and define the boundary edge set
\[
\partial K=\{e\in E : \sum_{t\supset e} z_{t,e}=1\}.
\]

\begin{lemma}[Combinatorial boundary identity]
For any selection satisfying (C1)--(C5),
\[
-\,W(B') = \sum_{e\in \partial K} L_e.
\]
\end{lemma}
\begin{proof}
If $e$ is incident to two selected triangles, it contributes $2L_e$ profit and incurs $2L_e$ cost, yielding net $0$.
If $e$ is incident to exactly one selected triangle, it contributes $L_e$ profit and incurs $2L_e$ cost, yielding net $-L_e$.
Summing over all edges gives the claim.
\end{proof}

\begin{lemma}[Soundness]
Any admissible solution defines an abstract triangulated disk whose boundary is a single simple Hamiltonian cycle.
\end{lemma}
\begin{proof}
By (C4), the selected triangles form a connected complex.
Constraints (C2) and (C5) exclude non-manifold edges and disconnected vertex links, so the complex is a simplicial surface with boundary.
By (C3), its Euler characteristic is $\chi=1$, hence it is a disk.
A disk has exactly one boundary component, so the boundary is a single simple cycle.
Moreover, (C5) implies every vertex lies on the boundary, so this cycle is Hamiltonian.
\end{proof}

\begin{theorem}[Equivalence]
$\mathrm{OPT}_{\mathrm{TSP}} = -\mathrm{OPT}_{\mathrm{cGSTP}}$.
\end{theorem}
\begin{proof}
\textbf{Soundness} follows from Lemma~2 and Lemma~1.

\textbf{Completeness.}
Let $C=(v_1,v_2,\dots,v_n)$ be any Hamiltonian cycle.
Select the $n-2$ triangles
\[
K:=\bigl\{\{v_1,v_i,v_{i+1}\}: i=2,\dots,n-1\bigr\}.
\]
Then $\partial K=C$.
Setting $x_t=1$ for $t\in K$, $y_e=1$ for all edges used by $K$, and $z_{t,e}=1$ for all incidences $e\subset t$ yields a feasible solution satisfying (C1)--(C5).
By Lemma~1, its objective value equals $-L(C)$.
Taking the optimum over $C$ gives $\mathrm{OPT}_{\mathrm{cGSTP}}\ge -\mathrm{OPT}_{\mathrm{TSP}}$, and combining with soundness yields equality.
\end{proof}

\section{Discussion}
In this note, we established a structural equivalence between the symmetric TSP and a constrained variant of the Group Steiner Tree Problem by reformulating tours as boundaries of admissible triangle selections.
The resulting model provides a unified constraint system in which global connectivity and objective cancellation arise naturally from the underlying simplicial incidence structure.
A key feature is that the construction is input-decoupled: it is exact when the full complex is available, and becomes a controlled heuristic when restricted to a prescribed candidate triangle set.
From a practical modeling perspective, the tour objective emerges through local cancellation of internal edges, while feasibility is enforced by a compact combination of a global tree constraint and local Euler regularity \cite{edelsbrunner2010}.
This makes it possible to use sparse geometric complexes (e.g., Delaunay or related triangulations) as black-box restrictions of the search space \cite{das1989,dillencourt1990,letchford2008}, while preserving exactness whenever an optimal tour is contained in the chosen complex.
Algorithmic development and a systematic computational study are left for future work.

\end{document}